\documentclass[a4paper,12pt]{article}

\usepackage{color}
\usepackage{amssymb, amsmath}
\usepackage{hyperref}
\usepackage{pifont}
\usepackage{tikz}% Pour dessiner un graphe
\usepackage{cite}

\newcommand\EFFACE[1]{}

\newtheorem{theorem}{Theorem}

\newtheorem{corollary}[theorem]{Corollary}

\newtheorem{observation}[theorem]{Observation}

\newenvironment{proof}{
\par
\noindent {\bf Proof.}\rm}{\mbox{}\hfill$\square$\par\vskip 3mm}

\makeatletter
\let\@fnsymbol\@arabic
\makeatother

\def\chicg{\chi_{\rm cg}}
\def\chig{\chi_{\rm g}}
\def\colg{{\rm col}_{\rm g}}
\def\colcg{{\rm col}_{\rm cg}}

\begin{document}

%\vspace*{2cm}
%\title{{\bf  A Connected Version of the Graph Coloring Game}} 
%\author{Cl\'ement CHARPENTIER\,\thanks{\texttt{clement.h.charpentier@gmail.com}}
%\and Herv\'e HOCQUARD\,\thanks{Univ. Bordeaux, LaBRI, UMR5800, F-33400 Talence, France 
%-- CNRS, LaBRI, UMR5800, F-33400 Talence, France. \texttt{(Herve.Hocquard, Eric.Sopena)@u-bordeaux.fr}}
%\and \'Eric SOPENA\,\footnotemark[2]
%\and Xuding ZHU\,\thanks{xxx}
%}
%\maketitle

\title{{\bf  A Connected Version of the Graph Coloring Game}} 
\author{Cl\'ement Charpentier\,\thanks{\texttt{clement.h.charpentier@gmail.com}}
\and Herv\'e Hocquard\,\thanks{Univ. Bordeaux, CNRS,  Bordeaux INP, LaBRI, UMR 5800, F-33400, Talence, France. \texttt{(Herve.Hocquard, Eric.Sopena)@u-bordeaux.fr}}
\and \'Eric Sopena\,\footnotemark[2]
\and Xuding Zhu\,\thanks{Department of Mathematics, Zhejiang Normal University, China. \texttt{xudingzhu@gmail.com}}
}
\maketitle

\abstract{
The graph coloring game is a two-player game in which,
given a graph~$G$ and a set of $k$ colors, the two players, Alice and Bob, 
take turns coloring properly an uncolored vertex of~$G$, Alice having the first move. 
Alice wins the game if and only if all the vertices of~$G$ are  colored. 
The game chromatic number of a graph $G$ is then defined as the smallest 
integer $k$ for which Alice has a winning strategy when playing the graph coloring game 
on~$G$ with $k$ colors.

In this paper, we introduce and study a new version of the graph coloring game
by requiring that the starting graph is connected and, after each player's turn, the subgraph induced by the set of colored vertices is connected.
The connected game chromatic number of a connected graph $G$ is then the smallest 
integer $k$ for which Alice has a winning strategy when playing the connected graph coloring game 
on~$G$ with $k$ colors.
We prove that the connected game chromatic number of every connected outerplanar graph is at most~5
and that there exist connected outerplanar graphs with connected game chromatic number~4.

Moreover, we prove that for every integer $k\ge 3$, there exist connected bipartite graphs on which
Bob wins the connected coloring game with $k$ colors,
while Alice wins the connected coloring game with two colors on every connected bipartite graph.
}%abstract



%%%%%%%%%%%%%%%%%%%%%%%%%%%%%%%%%%%%%%%%%%%%%%%%%%%%%%%%%%%%%%%%%%%%%%%%%%%%%%%%%%%%%%%%%%%%%%%%%%%%%%%%%%%%%%%%%%%%%%%%%%%%%
%%%%%%%%%%%%%%%%%%%%%%%%%%%%%%%%%%%%%%%%%%%%%%%%%%%%%%%%%%%%%%%%%%%%%%%%%%%%%%%%%%%%%%%%%%%%%%%%%%%%%%%%%%%%%%%%%%%%%%%%%%%%%
%%%%%%%%%%%%%%%%%%%%%%%%%%%%%%%%%%%%%%%%%%%%%%%%%%%%%%%%%%%%%%%%%%%%%%%%%%%%%%%%%%%%%%%%%%%%%%%%%%%%%%%%%%%%%%%%%%%%%%%%%%%%%
%%%%%%%%%%%%%%%%%%%%%%%%%%%%%%%%%%%%%%%%%%%%%%%%%%%%%%%%%%%%%%%%%%%%%%%%%%%%%%%%%%%%%%%%%%%%%%%%%%%%%%%%%%%%%%%%%%%%%%%%%%%%%
%%%%%%%%%%%%%%%%%%%%%%%%%%%%%%%%% INTRODUCTION
\section{Introduction}\label{sec:introduction}

All the graphs we consider in this paper are undirected, simple, and have no loops.
For every such graph $G$, we denote by $V(G)$ and $E(G)$ its vertex set and edge set, respectively,
by $\Delta(G)$ its maximum degree, and by $N_G(v)$ the set of neighbors of the vertex $v$ in $G$.

The {\em graph coloring game} is a two-player game introduced
by Steven J. Brams (reported by Martin Gardner in his column
{\it Mathematical Games} in {\it Scientific American} in 1981~\cite{G81}) 
and rediscovered ten years later by Bodlaender~\cite{B91}. 
Given a graph~$G$ and a set~${\cal C}$ of $k$ colors, the two players, Alice and Bob, 
take turns coloring properly an uncolored vertex of~$G$, Alice having the first move. 
Alice wins the game if and only if all the vertices of~$G$ are  colored. 
In other words,  Bob wins the game if and only if, at some step
of the game, all the colors appear in the neighborhood of some uncolored
vertex. 

The {\em game chromatic number} $\chig(G)$ of~$G$ is then defined as the smallest 
integer $k$ for which Alice has a winning strategy when playing the graph coloring game 
on~$G$ with $k$ colors.
The problem of determining the game chromatic number of several graph
classes has attracted much interest in recent years
(see~\cite{BGKZ07} for a comprehensive survey of this problem),
with a particular focus on planar graphs (see e.g. \cite{DZ99,K00,KT94,WZ08,Z99,Z08})
for which the best known upper bound up to now is 17~\cite{Z08}.
Very recently, Costa, Pessoa, Sampaio and Soares~\cite{CPSS19} proved that
given a graph $G$ and an integer $k$, deciding
whether $\chig(G)\le k$ is a PSPACE-Complete problem, thus answering a longstanding
open question. 

Most of the known upper bounds on the game chromatic number of classes of graphs 
are derived from upper bounds on the game coloring number of these classes,
a parameter defined through the so-called {\em graph marking game},
formally introduced by Zhu in~\cite{Z99}.
This game is somehow similar to the graph coloring game, except that the players
mark the vertices instead of coloring them, with no restriction.
The {\em game coloring number} $\colg(G)$ of~$G$ is then defined 
as the smallest integer $k$ for which Alice has a strategy such that,
when playing the graph marking game on~$G$, every unmarked vertex has at most
$k-1$ marked neighbors.
It is worth noting here that the game coloring number is \emph{monotonic},
which means that $\colg(H)\le\colg(G)$ for every subgraph $H$ of $G$,
while this property does not hold for the game chromatic number~\cite{W05}.

Let $G$ be a graph with $\colg(G)=k$ and consider the winning strategy of Alice for the marking
game on~$G$. Applying the same strategy for the coloring game on~$G$, Alice ensures
that each uncolored vertex has at most $k-1$ colored neighbors, so that we
get $\chig(G)\le k$. Hence, the following inequalities clearly hold
for every graph~$G$.

\begin{observation}\label{obs:inequality-chig}
For every graph $G$, $\chi(G)\le \chig(G)\le\colg(G)\le \Delta(G)+1$.
\end{observation}

\medskip

In this paper, we introduce and study a new version of the graph coloring game
(resp. of the graph marking game), played on connected graphs, by requiring
that, after each player's turn, the subgraph induced by the set of colored 
(resp. marked) vertices is connected.
In other words, on their turn, each player must color an uncolored vertex
(resp. mark an unmarked vertex) having at
least one colored (resp. marked) neighbor, except for Alice on her first move.

We call this new game the {\em connected graph coloring game}
(resp. the {\em connected graph marking game}).
We will denote by $\chicg(G)$
the {\em connected game chromatic number} of a connected graph $G$, that is, the smallest 
integer $k$ for which Alice has a winning strategy when playing the connected graph coloring game on
$G$ with $k$ colors,
and by $\colcg(G)$
the {\em connected game coloring number} of $G$, that is, the smallest 
integer $k$ for which Alice has a strategy such that,
when playing the connected graph marking game on~$G$, every unmarked vertex has at most
$k-1$ marked neighbors.
It is not difficult to observe that, similarly to the ordinary case, the following inequalities
hold for every connected graph~$G$.

\begin{observation}\label{obs:inequality-chicg}
For every connected graph $G$, $\chi(G)\le \chicg(G)\le\colcg(G)\le \Delta(G)+1$.
\end{observation}

%\ES{However, the relation between $\chig(G)$ and $\chicg(G)$, or between $\colg(G)$ and $\colcg(G)$,
%is not so clear...} 

It is proved in~\cite{W05} that for every positive integer $n$, $\chig(K_{n,n}-M)=n$,
where $K_{n,n}-M$ denotes the complete bipartite graph with $n$ vertices in each part, minus a perfect matching.
We prove in Section~\ref{sec:bipartite} that $\chicg(G)=2$ for every nonempty connected bipartite graph $G$,
which shows, since the graph $K_{n,n}-M$ is bipartite, 
that the difference $\chig(G)-\chicg(G)$ can be arbitrarily large.
%However, \ES{do we have an example of a graph $G$ with $\chicg(G) > \chig(G)$?}

\medskip

One of the main open, and rather intriguing, question concerning the graph coloring
game is the following:
assuming that Alice has a winning strategy for the graph coloring game on a
graph $G$ with $k$ colors, is it true that she has also a winning strategy with $k+1$ colors?
We will prove in Section~\ref{sec:bipartite} that the answer is ``no'' for the connected
version of the coloring game.
More precisely, we will prove that 
for every integer $k\ge 3$, there exist connected bipartite graphs on which
Bob wins the connected coloring game with $k$ colors,
while Alice wins the connected coloring game with two colors on every connected bipartite graph.

The ``connected variant'' of other types of games on graphs have been considered in the
literature. This is the case for instance for
the domination game~\cite{BFS,I19},
the surveillance game~\cite{FGJMN14,GLMNP15},
the graph searching game~\cite{BFFFNST12,BGTZ16,FN08},
or Hajnal's triangle-free game~\cite{MS11,S92}.
However, to our knowledge, the connected variant of the graph
coloring game has not been considered yet.

\medskip

Our paper is organized as follows. 
We consider bipartite graphs in Section~\ref{sec:bipartite}, and 
outerplanar graphs in Section~\ref{sec:outerplanar}.
We finally propose some directions for future research in Section~\ref{sec:discussion}.

%%%%%%%%%%%%%%%%%%%%%%%%%%%%%%%%%%%%%%%%%%%%%%%%%%%%%%%%%%%%%%%%%%%%%%
%%%%%%%%%%%%%%%%%%%%%%%%%%%%%%%%%%%%%%%%%%%%%%%%%%%%%%%%%%%%%%%%%%%%%%
%%%%%%%%%%%%%%%%%%%%%%%%%%%%%%%%%%%%%%%%%%%%%%%%%%%%%%%%%%%%%%%%%%%%%%
%%%%%%%%%%%%%%%%%%%%%%%%%%%%%%%%%%%%%%%%%%%%%%%%%%%%%%%%%%%%%%%%%%%%%%
%%%%%%%%%%%%%%%%%%%%%%%%%%%%%%%%%%%%%%%%%%%%%%%%%%%%%%%%%%%%%%%%%%%%%%
%%%%%%%%%%%%%%%%%%%%%%%%%%%%%%%%%%%%%%%%%%%%%%%%%%%%%%%%%%%%%%%%%%%%%%

\section{Bipartite graphs}\label{sec:bipartite}

We consider the case of bipartite graphs in this section.
We will prove that 
for every integer $k\ge 3$, there exist connected bipartite graphs on which
Bob wins the connected coloring game with $k$ colors,
while Alice wins the connected coloring game with two colors on every bipartite graph.

It is easy to see that Alice always wins when playing
the connected coloring game on a connected bipartite graph $G$ with two colors:
thanks to the connectivity constraint, the first move of Alice 
forces all the next moves to be consistent with a proper 2-coloring
of~$G$.

\begin{theorem}\label{th:bipartite}
For every connected bipartite graph $G$, $\chicg(G)\le 2$.
\end{theorem}

\begin{proof}
Let $G$ be any connected bipartite graph.
The strategy of Alice is as follows.
On her first move, she picks any vertex $v$ of $G$ and gives it color~1.
From now on, each play will color some vertex having at least one of its
neighbors already colored, so that, since $G$ is bipartite, this eventually
leads to a proper 2-coloring of $G$.
\end{proof}

However, for every integer $k\ge 3$, there are connected bipartite graphs on which Bob
wins the connected coloring game with $k$~colors.

\begin{theorem}\label{th:bipartite-Bob}
For every integer $k\ge 3$, there exists a connected bipartite graph $G_k$ on which
Bob wins the connected coloring game with $k$~colors.
\end{theorem}

\begin{proof}
Let $H_k$ be any $C_4$-free bipartite graph with minimum degree at least $k^2$ and let $A$ and $B$ denote the partite sets of $H_k$. (Consider for instance the incidence graph of a projective plane
of dimension $d\ge k^2$; such a graph is a $(d+1)$-regular bipartite graph with girth 6.)
Let now $G_k$ be the (connected) bipartite graph obtained from $H_k$ by adding, for each $k$-subset $S$ of $B$, a new vertex
$v_S$ adjacent to all vertices of $S$.

We now define the strategy of Bob for playing the connected coloring game on $G_k$ as follows.
In his first moves (at most three, depending on the moves of Alice), Bob colors two vertices of $A$,
say $u$ and $v$, with two different colors. 
In his next two moves, Bob colors a neighbor $u'$ of $u$ in $B$ with the same color as $v$,
and a neighbor $v'$ of $v$ in $B$ with the same color as $u$.
Since the minimum degree of $H_k$ is at least $k^2$ and $H_k$ is $C_4$-free, Alice cannot prevent Bob from doing so.

Now, Bob colors a $k$-subset $X \subseteq N_{H_k}(u) \cup N_{H_k}(v)$ containing $u'$ and $v'$
with $k$ distinct colors. Again, Alice cannot prevent Bob from doing so since each move
of Alice ``forbids'' at most $k$ uncoloured vertices in $N_{H_k}(u) \cup N_{H_k}(v)$
(each vertex of $G_k$ has at most $k$ neighbors in this set).

After that, the vertex $v_X$ cannot be colored and Bob wins the game.
\end{proof}

%%%%%%%%%%%%%%%%%%%%%%%%%%%%%%%%%%%%%%%%%%%%%%%%%%%%%%%%%%%%%%%%%%%%%%%%%%%%%%%%%%%
%%%%%%%%%%%%%%%%%%%%%%%%%%%%%%%%%%%%%%%%%%%%%%%%%%%%%%%%%%%%%%%%%%%%%%%%%%%%%%%%%%%
%%%%%%%%%%%%%%%%%%%%%%%%%%%%%%%%%%%%%%%%%%%%%%%%%%%%%%%%%%%%%%%%%%%%%%%%%%%%%%%%%%%
%%%%%%%%%%%%%%%%%%%%%%%%%%%%%%%%%%%%%%%%%%%%%%%%%%%%%%%%%%%%%%%%%%%%%%%%%%%%%%%%%%%

\section{Outerplanar graphs}\label{sec:outerplanar}

We consider in this section the case of outerplanar graphs.
An \emph{outerplanar graph} is a graph that can be
embedded on the plane in such a way that there are no edge crossings and all its vertices lie on 
the outer face. Recall that a graph is outerplanar if and only if it does not contain
$K_4$ or $K_{2,3}$ as a minor.

Concerning the ordinary coloring game,
Kierstead and Trotter proved in~\cite{KT94} that there
exist outerplanar graphs with game chromatic number at least~6,
and Guan and Zhu proved in~\cite{GZ99} that the game chromatic
number of every outerplanar graph is at most~7.
In~\cite{KY05} Kierstead and Yang have proved that the bound is tight
for the game coloring number of outerplanar graphs.
%This bound has then been proven to be tight by Kierstead and Yang in~\cite{KY05}.
We will prove that the connected game chromatic
number of every outerplanar graph is at most~5 and that there exist
outerplanar graphs with connected game chromatic number~4.

%An \emph{outerplane graph} is an outerplanar graph embedded in such a way.
%Slightly abusing notation, we will generally use the same symbol to denote an
%outerplane graph and the corresponding outerplanar graph.
%

\medskip

Recall that an outerplanar graph is \emph{maximal} if adding any edge makes it non outerplanar.
Therefore, an outerplanar graph is maximal if and only if, in all its
outerplanar embeddings, all faces are triangles, except possibly the outer face
(an outerplanar graph is thus a triangulation of a plane cycle).
In particular, every maximal outerplanar graph is connected.
Our goal in this section is to prove that the connected coloring number of 
every connected outerplanar graph is at most~5.

When playing the connected coloring game
on a connected graph~$G$, we will say that an uncolored vertex in~$G$
is \emph{saturated}
if each of the available colors appears in its neighborhood. 
Observe that Bob wins
the connected coloring game on $G$ if and only if he has a strategy such that,
at some point in the game, an uncolored vertex in $G$ becomes saturated.
Similarly, when trying to prove that the connected game coloring number
of some graph~$G$ is at most $k$, we will say that an unmarked vertex in~$G$
is saturated if it has at least $k$ marked neighbors.
Again, the connected game coloring number of~$G$ is at least $k+1$ if and only if
Bob has a strategy such that,
at some point in the game, an unmarked vertex in $G$ becomes saturated.

Finally, we will say that a vertex in $G$ is 
{\em playable} if it is uncolored (resp. unmarked) and has at least one colored (resp. marked) neighbor.
Moreover, when considering the connected marking game, we will say that a vertex
is {\em threatened} if it is unmarked, 
has $k-1$  marked neighbors and at least one playable neighbor.
In that case, note that if Bob plays on a playable neighbor of any threatened vertex, then Alice loses the game.
A winning strategy of Alice for the connected marking game
 thus consists in ensuring that, after each of her moves,
the considered graph has no threatened vertex.

\medskip

\begin{figure}
\begin{center}
\begin{tikzpicture}[x=1cm,y=1cm]
\node[left] (V0) at (-1,2) {$V_0$}; \draw[->,dashed] (V0) -- (1.7,2);
\node[left] (V1) at (-1,0) {$V_1$}; \draw[->,dashed] (V1) -- (-0.3,0);
\node[left] (V2) at (-1,-2) {$V_2$}; \draw[->,dashed] (V2) -- (-0.3,-2);
\node[scale=0.7,draw,circle,fill=black] (u) at (2,2) {}; \node[above] at (2,2.2) {$u$};
\node[scale=0.7,draw,circle,fill=black] (up) at (3,2) {}; \node[above] at (3,2.2) {$u'$};
\node[scale=0.7,draw,circle,fill=black] (v0) at (0,0) {};
\node[scale=0.7,draw,circle,fill=black] (v1) at (1,0) {};
\node[scale=0.7,draw,circle,fill=black] (v2) at (3,0) {};
\node[scale=0.7,draw,circle,fill=black] (v3) at (4,0) {};
\node[scale=0.7,draw,circle,fill=black] (v4) at (5,0) {};
\draw[-] (u) -- (up);
\draw[-] (u) -- (v0);
\draw[-] (v0) -- (v1);
\draw[-,very thick,dotted] (v1) -- (v2);
\draw[-] (v2) -- (v4);
\draw[-] (v4) -- (up);
\node[scale=0.7,draw,circle,fill=black] (w0) at (0,-2) {};
\node[scale=0.7,draw,circle,fill=black] (w1) at (1.3,-2) {};
\node[scale=0.7,draw,circle,fill=black] (w2) at (2.7,-2) {};
\node[scale=0.7,draw,circle,fill=black] (w3) at (4,-2) {};
\node[scale=0.7,draw,circle,fill=black] (w4) at (5.2,-2) {};
\node[scale=0.7,draw,circle,fill=black] (w5) at (6.4,-2) {};
\draw[-] (v0) -- (w0);
\draw[-] (v1) -- (w1);
\draw[-] (v2) -- (w2);
\draw[-] (v3) -- (w3);
\draw[-] (v3) -- (w4);
\draw[-] (v4) -- (w5);
\draw[-,very thick,dotted] (w0) -- (w1);
\draw[-,very thick,dotted] (w2) -- (w3);
\draw[-,very thick,dotted] (w4) -- (w5);
\node[right] at (0,-2.7) {{\it Etc.}};
\node at (2.7,-3.5) {(a)};
\end{tikzpicture}
\hskip 1cm
\begin{tikzpicture}[x=1cm,y=1cm]
\node[scale=0.7,draw,circle,fill=black] (v) at (2,2) {}; \node[above] at (2,2.2) {$v$};
\node[scale=0.7,draw,circle,fill=black] (vprime) at (3,2) {}; \node[above] at (3,2.2) {$v'$};
\node[scale=0.7,draw,circle,fill=black] (v1) at (-1,0) {}; \node[below] at (-1,-0.2) {$v_1$};
\node[scale=0.7,draw,circle,fill=black] (v2) at (0,0) {}; \node[below] at (0,-0.2) {$v_2$};
\node[scale=0.7,draw,circle,fill=black] (vp) at (2,0) {}; \node[below] at (2,-0.2) {$v_p$};
\node[scale=0.7,draw,circle,fill=black] (vk1) at (4,0) {}; \node[below] at (4,-0.2) {$v_{k-1}$};
\node[scale=0.7,draw,circle,fill=black] (vk) at (5,0) {}; \node[below] at (5,-0.2) {$v_k$};
%\node[scale=0.7,draw,circle,fill=black] (v5) at (5,0) {};
\draw[-] (v) -- (vprime);
\draw[-] (v) -- (v1);
\draw[-] (v) -- (v2);
\draw[-] (v) -- (vp);
\draw[-] (v1) -- (v2);
\draw[-,very thick,dotted] (v2) -- (vk1);
\draw[-] (vk1) -- (vk);
\draw[-] (vprime) -- (vp);
\draw[-] (vprime) -- (vk1);
\draw[-] (vprime) -- (vk);
\node at (2,-2.5) {(b)};
\end{tikzpicture}

\caption{\label{fig:out-structure}The structure of a maximal outerplanar graph.}
\end{center}
\end{figure}
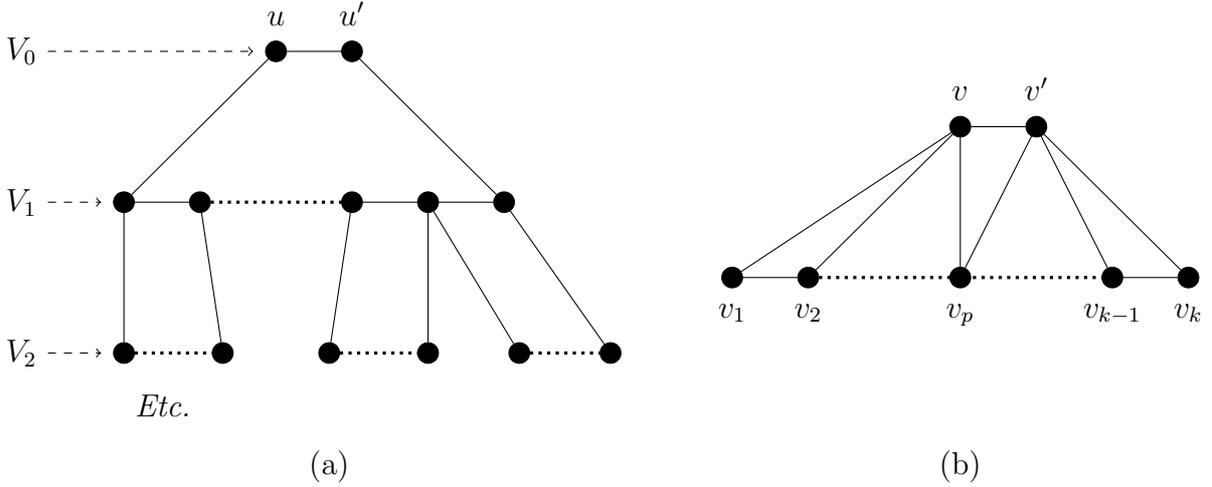

We first describe more precisely the structure of maximal outerplanar graphs, which will
be used for defining the strategy of Alice.
The structure of maximal outerplanar graphs has been studied by several authors
(see e.g.~\cite{A09,HPS85,LMN12,M79}, just to cite a few).
In particular, the neighborhood of every vertex in a maximal outerplanar graph induces a path.
Moreover, the {\em triangle graph} $T(G)$ of a maximal (embedded) outerplanar graph $G$,
whose vertices are the triangle faces of $G$, incident faces being linked by an edge,
is necessarily a tree.

Let $G$ be a maximal (embedded) outerplanar graph.
An edge belonging to the outer face of $G$ is an \emph{outer edge} of $G$.
%An edge which is not an outer edge is a \emph{chord}.
Let us choose and fix any outer edge $e=uu'$ of $G$. 
The \emph{distance} from any vertex $v$ to the edge $e$ is defined as $d_G(v,e)=\min\left\{d_G(v,u),d_G(v,u')\right\}$.
For every integer $i\ge 0$, let $V_i$ denote the set of vertices at distance $i$ from $e$.
Observe that the subgraph $G[V_i]$ of $G$ induced by each set $V_i$ is a linear forest, that is,
a disjoint union of paths, since otherwise $G$ would contain a $K_{2,3}$ as a minor. 
In particular, $G[V_0]$ is the edge $uu'$ and $G[V_1]$ is a path.
Therefore, $G$ can be viewed as a ``tree of trapezoids'', as illustrated
in Figure~\ref{fig:out-structure}(a).

Each of these trapezoids has the structure depicted in Figure~\ref{fig:out-structure}(b).
Both $v$ and $v'$ belong to some $V_i$, $i\ge 0$, 
$vv'$ being an edge-cut of $G$,
while $v_1,\dots,v_k$, $k\ge 2$, belong to $V_{i+1}$.
The vertices $v$ and $v'$ are the \emph{parents} of the \emph{children vertices} $v_1,\dots,v_k$,
the edge $vv'$ is the \emph{root edge} of the trapezoid, 
the unique vertex $v_p$, $1\le p\le k$, which is joined by an edge to both $v$ and $v'$
is the \emph{pivot} of the trapezoid
(its uniqueness comes from the fact that $G$ does not contain $K_4$ as a minor).
We will denote by $T_{vv'}$ the trapezoid whose root edge is $vv'$, and by $p(vv')$
the pivot of $T_{vv'}$.
Note that each vertex $v_i$, $1\le i\le k$, is a neighbor of at least one of its parents,
and that only the pivot $v_p$ is a neighbor of both its parents.
Moreover, we will say that $v_{i-1}$, $2\le i\le k$, is the \emph{left neighbor} of $v_i$,
while $v_{i+1}$, $1\le i\le k-1$, is the \emph{right neighbor} of $v_i$
(this ordering is well defined since the embedding of $G$ is given).
%Similarly, if $T_{v_iv_{i+1}}$ and $T_{v_jv_{j+1}}$ are two trapezoids
%with $1\le i<j<k$, we will say that $T_{v_iv_{i+1}}$ is \emph{to the left of} $T_{v_jv_{j+1}}$,
%and that $T_{v_jv_{j+1}}$ is \emph{to the right of} $T_{v_iv_{i+1}}$.
%

%We will denote by $T_{vv'}$ the trapezoid whose root edge is $vv'$, and by $p(vv')$
%the pivot of $T_{vv'}$.

Observe that if there is no trapezoid of the form $T_{vv'}$ or $T_{v'v}$ in $G$ for some vertex $v$,
then the degree of $v$ is at most~4 (it is~4 only if $v$ is the pivot of some trapezoid),
so that $v$ cannot be a threatened vertex.
Note also that every vertex belongs to at most two root edges.

Based on the drawing of the outerplanar graph depicted in Figure~\ref{fig:out-structure}(a),
we can define a total ordering $\le_G$ of the vertices of $G$,
obtained by listing the vertices of $V_0$ from left to right, then 
the vertices of $V_1$ from left to right, and so on.
Finally, we will say that a vertex $w_1$, belonging to a trapezoid $T_{v_1v'_1}$,
\emph{lies above} a vertex $w_2$, belonging to a trapezoid $T_{v_2v'_2}$,
if every shortest path from $\{v_2,v'_2\}$ to $\{u,u'\}$ goes through $v_1$ or $v'_1$.

%Finally, we define the \emph{trapezoidal tree of} $G$, denoted $T(G)$,
%as the tree whose vertices are the trapezoids of $G$, and edges are all
%the pairs $\{T_{v_1v'_1},T_{v_2v'_2})$ such that $v_2$ and $v'_2$ belong to $T_{v_1v'_1}$.
%We will then say that a vertex $w_1$, belonging to a trapezoid $T_{v_1v'_1}$,
%\emph{lies above} a vertex $w_2$, belonging to a trapezoid $T_{v_2v'_2}$,
%if $T_{v_1v'_1}$ belongs to the path joining $T_{uu'}$ and $T_{v_2v'_2}$ in $T(G)$.

\medskip

\newcommand\TRAPEZE[2]{
\draw[-] (#1,#2) -- (#1+3,#2);
\draw[-] (#1,#2) -- (#1+1,#2+2);
\draw[-] (#1+1,#2+2) -- (#1+2,#2+2);
\draw[-] (#1+2,#2+2) -- (#1+3,#2);
}
\newcommand\TRAPEZEDASHED[2]{
\draw[-] (#1,#2) -- (#1+3,#2);
\draw[-] (#1,#2) -- (#1+1,#2+2);
\draw[-,dashed] (#1+1,#2+2) -- (#1+2,#2+2);
\draw[-] (#1+2,#2+2) -- (#1+3,#2);
}

\begin{figure}
\begin{center}
\begin{tikzpicture}[x=0.8cm,y=0.8cm]
\TRAPEZEDASHED{0}{0}
\node[scale=0.7,draw,circle,fill=black] (v) at (1,2) {};
\node[scale=0.7,draw,circle,fill=white] (ai) at (2,2) {};
\node[scale=0.7,draw,circle,fill=black] (bi) at (1.2,0) {}; \node[below] at (1.2,-0.2) {$b_i$};
\draw[-] (v) -- (bi);
    \draw[->] (3.5,1) -- (4.5,1);
\TRAPEZEDASHED{5+0}{0}
\node[scale=0.7,draw,circle,fill=black] (vv) at (5+1,2) {};
\node[scale=0.7,draw,circle,fill=black] (aai) at (5+2,2) {}; \node[above] at (5+2,2.2) {$a_i$};
\node[scale=0.7,draw,circle,fill=black] (bbi) at (5+1.2,0) {}; \node[below] at (5+1.2,-0.2) {$b_i$};
\draw[-] (vv) -- (bbi);
%
%\node at (-0.6,1) {R1:};
\end{tikzpicture}
\hskip 1.5cm
\begin{tikzpicture}[x=0.8cm,y=0.8cm]
\TRAPEZEDASHED{0}{0}
\node[scale=0.7,draw,circle,fill=white] (ai) at (1,2) {};
\node[scale=0.7,draw,circle,fill=black] (v) at (2,2) {};
\node[scale=0.7,draw,circle,fill=black] (bi) at (1.2,0) {}; \node[below] at (1.2,-0.2) {$b_i$};
\draw[-] (v) -- (bi);
    \draw[->] (3.5,1) -- (4.5,1);
\TRAPEZEDASHED{5+0}{0}
\node[scale=0.7,draw,circle,fill=black] (aai) at (5+1,2) {}; \node[above] at (5+1,2.2) {$a_i$};
\node[scale=0.7,draw,circle,fill=black] (vv) at (5+2,2) {}; 
\node[scale=0.7,draw,circle,fill=black] (bbi) at (5+1.2,0) {}; \node[below] at (5+1.2,-0.2) {$b_i$};
\draw[-] (vv) -- (bbi);
%
%\node at (-0.6,1) {R1:};
\end{tikzpicture}
\vskip 0.1cm
Rule R1 (the vertex $a_i$ must be playable)
\vskip 0.7cm
\begin{tikzpicture}[x=0.8cm,y=0.8cm]
\TRAPEZEDASHED{0}{0}
\node[scale=0.7,draw,circle,fill=black] (v) at (1,2) {};
\node[scale=0.7,draw,circle,fill=black] (bi) at (2,2) {}; \node[above] at (2,2.2) {$b_i$};
\node[scale=0.7,draw,circle,fill=white] (ai) at (1.5,0) {}; %\node[below] at (1.5,-0.2) {$a_i$};
\draw[-,dashed] (v) -- (ai);
\draw[-,dashed] (ai) -- (bi);
    \draw[->] (3.5,1) -- (4.5,1);
\TRAPEZEDASHED{5+0}{0}
\node[scale=0.7,draw,circle,fill=black] (vv) at (5+1,2) {};
\node[scale=0.7,draw,circle,fill=black] (bbi) at (5+2,2) {}; \node[above] at (5+2,2.2) {$b_i$};
\node[scale=0.7,draw,circle,fill=black] (aai) at (5+1.5,0) {}; \node[below] at (5+1.5,-0.2) {$a_i$};
\draw[-,dashed] (vv) -- (aai);
\draw[-,dashed] (aai) -- (bbi);
%
%\node at (-0.6,1) {R2:};
\end{tikzpicture}
\hskip 1.5cm
\begin{tikzpicture}[x=0.8cm,y=0.8cm]
\TRAPEZEDASHED{0}{0}
\node[scale=0.7,draw,circle,fill=black] (v) at (1,2) {};
\node[scale=0.7,draw,circle,fill=black] (bi) at (2,2) {}; \node[above] at (1,2.2) {$b_i$};
\node[scale=0.7,draw,circle,fill=white] (ai) at (1.5,0) {}; %\node[below] at (1.5,-0.2) {$a_i$};
\draw[-,dashed] (v) -- (ai);
\draw[-,dashed] (ai) -- (bi);
    \draw[->] (3.5,1) -- (4.5,1);
\TRAPEZEDASHED{5+0}{0}
\node[scale=0.7,draw,circle,fill=black] (vv) at (5+1,2) {};
\node[scale=0.7,draw,circle,fill=black] (bbi) at (5+2,2) {}; \node[above] at (5+1,2.2) {$b_i$};
\node[scale=0.7,draw,circle,fill=black] (aai) at (5+1.5,0) {}; \node[below] at (5+1.5,-0.2) {$a_i$};
\draw[-,dashed] (vv) -- (aai);
\draw[-,dashed] (aai) -- (bbi);
%
%\node at (-0.6,1) {R2:};
\end{tikzpicture}
\vskip 0.1cm
Rule R2 ($a_i$ is a pivot and at least one of the dashed edges incident with $a_i$ must exist)
\vskip 0.7cm
\begin{tikzpicture}[x=0.8cm,y=0.8cm]
\TRAPEZEDASHED{0}{0}
\node[scale=0.7,draw,circle,fill=white] (ai) at (1,2) {};
\node[scale=0.7,draw,circle,fill=black] (bi) at (2,2) {}; \node[above] at (2,2.2) {$b_i$};
\draw[-] (ai) -- (0.6,2.7); \draw[-] (ai) -- (1.4,2.7);
    \draw[->] (3.5,1) -- (4.5,1);
\TRAPEZEDASHED{5+0}{0}
\node[scale=0.7,draw,circle,fill=black] (aai) at (5+1,2) {}; \node[above] at (5+1,2.2) {$a_i$};
\node[scale=0.7,draw,circle,fill=black] (bbi) at (5+2,2) {}; \node[above] at (5+2,2.2) {$b_i$};
\draw[-] (aai) -- (5+0.6,2.7); \draw[-] (aai) -- (5+1.4,2.7);
%
%\node at (-0.6,1) {R3:};
\end{tikzpicture}
\hskip 1.5cm
\begin{tikzpicture}[x=0.8cm,y=0.8cm]
\TRAPEZEDASHED{0}{0}
\node[scale=0.7,draw,circle,fill=black] (bi) at (1,2) {}; \node[above] at (1,2.2) {$b_i$};
\node[scale=0.7,draw,circle,fill=white] (ai) at (2,2) {}; 
\draw[-] (ai) -- (1+0.6,2.7); \draw[-] (ai) -- (1+1.4,2.7);
    \draw[->] (3.5,1) -- (4.5,1);
\TRAPEZEDASHED{5+0}{0}
\node[scale=0.7,draw,circle,fill=black] (bbi) at (5+1,2) {}; \node[above] at (5+1,2.2) {$b_i$};
\node[scale=0.7,draw,circle,fill=black] (aai) at (5+2,2) {}; \node[above] at (5+2,2.2) {$a_i$};
\draw[-] (aai) -- (1+5+0.6,2.7); \draw[-] (aai) -- (1+5+1.4,2.7);
%
%\node at (-0.6,1) {R3:};
\end{tikzpicture}
\vskip 0.3cm
Rule R3 (the vertex $a_i$ must be playable)
%%
%\vskip 1cm
%%
%\begin{tikzpicture}[x=0.8cm,y=0.8cm]
%\TRAPEZE{0}{0}
%\node[scale=0.7,draw,circle,fill=white] (ai) at (1,2) {};
%\node[scale=0.7,draw,circle,fill=black] (bi) at (2,2) {}; \node[above] at (2,2.2) {$b_i$};
%    \draw[->] (3.5,1) -- (4.5,1);
%\TRAPEZE{5+0}{0}
%\node[scale=0.7,draw,circle,fill=black] (aai) at (5+1,2) {}; \node[above] at (5+1,2.2) {$a_i$};
%\node[scale=0.7,draw,circle,fill=black] (bbi) at (5+2,2) {}; \node[above] at (5+2,2.2) {$b_i$};
%%
%\node at (-0.6,1) {R4:};
%\end{tikzpicture}
%%
%\hskip 1cm
%%
%\begin{tikzpicture}[x=0.8cm,y=0.8cm]
%\TRAPEZE{0}{0}
%\node[scale=0.7,draw,circle,fill=white] (ai) at (2,2) {}; \node[above] at (1,2.2) {$b_i$};
%\node[scale=0.7,draw,circle,fill=black] (bi) at (1,2) {}; 
%    \draw[->] (3.5,1) -- (4.5,1);
%\TRAPEZE{5+0}{0}
%\node[scale=0.7,draw,circle,fill=black] (aai) at (5+2,2) {}; \node[above] at (5+1,2.2) {$b_i$};
%\node[scale=0.7,draw,circle,fill=black] (bbi) at (5+1,2) {}; \node[above] at (5+2,2.2) {$a_i$};
%%
%\node at (-0.6,1) {R5:};
%\end{tikzpicture}
\caption{\label{fig:out-strategy}The strategy of Alice on outerplanar graphs (Rules R1, R2 and R3).}
\end{center}
\end{figure}

We now describe the strategy of Alice when playing the connected marking game
on a connected outerplanar graph $G$. 
Let $uu'$ be any outer edge of $G$, and $G_m$ be any maximal outerplanar graph
containing $G$ as a subgraph, and such that $uu'$ is also an outer edge of $G_m$.
In the following, we assume that we are given a trapezoidal representation of $G_m$,
starting from the edge $uu'$, as described above. Moreover, we can also assume that
for every trapezoid $T_{vv'}$ of $G_m$, the pivot $p(vv')$ has been chosen
in such a way that it is linked by an edge in $G$ to at least one vertex from $\{v,v'\}$.
%
%at least one of the edges $vp(vv')$,
%$v'p(vv')$ belongs to $G$.
This will allow us to speak about children or parent vertices (with respect to $G_m$)
even if the corresponding edges do not belong to $G$, and to use the ordering $\le_{G_m}$
of the vertices of $G$.

Let us denote by $a_i$, $i\ge 0$, the vertex marked by Alice on her $(i+1)$-th move,
and by $b_i$, $i\ge 1$, the  vertex marked by Bob on his $i$-th move, so that the
sequence of moves (that is, marked vertices) is  $a_0,b_1,a_1,\dots,b_i,a_i,\dots$
Hence, $a_0$ is the vertex marked by Alice on her first move and, for every $i\ge 1$,
$a_i$ is the ``response'' of Alice to the move $b_i$ of Bob.
%Let us denote by $b$ the vertex marked by Bob just before Alice's turn.

The strategy of Alice will then consist in applying the first of the 
following rules that can be applied (see Figure~\ref{fig:out-strategy} for an illustration
of Rules R1, R2 and R3) for each of her moves. 

\begin{itemize}
\item[R0:] $a_0:=u$.
\item[R1:] If $v$ is a playable unmarked parent of $b_i$, then $a_i:=v$.
\item[R2:] If $b_i$ belongs to a root edge $vb_i$ or $b_iv$, $v$ is marked
and $p(vb_i)$ is playable, then $a_i:=p(vb_i)$.
\item[R3:] If $b_i$ belongs to a root edge $vb_i$ or $b_iv$, $v$ is unmarked,  $v$ is a pivot
and $v$ is playable, then $a_i:=v$.
\item[R4:] If none of the above rules can be applied, and there are still unmarked
vertices in $G$, then we let $a_i:=w$, where $w$ is the smallest 
(with respect to the ordering $\le_{G_m}$) playable vertex.
%%
%\item[R4:] If $b_i$ belongs to a root edge $vb_i$ and $v$ is unmarked, then $a_i:=v$.
%%
%\item[R5:] If $b_i$ belongs to a root edge $b_iv$ and $v$ is unmarked, then $a_i:=v$.
%
%\item[R4:] If none of the above rules can be applied, and there are still unmarked
%vertices in $G$, then 
%%$b_i$ does not belong to any root edge, and both parents of $b_i$ are marked.
%%In that case, 
%we let $a_i:=w$, where $w$ is the vertex defined as follows:
%If a marked vertex has a playable left or right neighbor, then we let $w$ be the smallest 
%(with respect to the ordering $\le_{G_m}$) such neighbor.
%Otherwise, we let $w$ be the smallest (with respect to the ordering $\le_{G_m}$) playable pivot.
%
%(Clearly, since every trapezoid has a pivot, if $G$ still contains unmarked vertices, 
%then at least one of these two cases occurs.)
%
\end{itemize}

Note that on his first move, Bob must mark either the vertex $u'$, in which case Alice will apply Rule R2
on her second move,
or some neighbor $v\neq u'$ of $u$, in which case Alice will apply rule R1 and mark $u'$ on her second move
(recall that the edge $uu'$ belongs to $G$).
Moreover, if Bob marks a child vertex $w$ of some trapezoid $T_{vv'}$,
then at least one of $v$, $v'$ must be marked (by the connectivity constraint), and Alice
will immediately apply Rule R1 if one of them is unmarked and $vv'$ is an edge in $G$.
%
%Note also that, by Rule R4, %Alice never marks a vertex having an unmarked parent and,
%if Alice marks a pivot $p(vv')$ whose left and right neighbors are unmarked, then this pivot
%is the only marked vertex lying below $v$ and~$v'$.
%
These remarks are summarized in the two following observations.

\begin{observation}\label{obs:out-uu'}
After the second move of Alice, both vertices $u$ and $u'$ are marked.
\end{observation}

\begin{observation}\label{obs:out-both-parents}
After each move of Alice, if $w$ is a marked child vertex of a trapezoid $T_{vv'}$
and $vv'$ is an edge in $G$,
then both $v$ and $v'$ are marked.
\end{observation}

%\begin{observation}\label{obs:out-Alice-pivot}
%If Alice uses Rule R4 and marks some pivot $p(vv')$ whose left and right neighbors are unmarked, 
%then the trapezoid $T_{vv'}$ contains at most three marked vertices, namely $v$, $v'$ and $p(vv')$,
%and no vertex lying below $v$ and $v'$, except $p(vv')$, is marked.
%\end{observation}

We are now able to prove the main result of this section.

\begin{theorem}\label{th:out-marking}
If $G$ is a connected outerplanar graph, then $\colcg(G)\le 5$.
\end{theorem}

\begin{proof}
We assume that we are given an outerplanar embedding of $G_m$ and 
its trapezoidal representation, as previously discussed.
Clearly, it suffices to prove that if Alice applies the above described strategy, then,
after each move of Alice, $G$ contains no threatened vertex.
This is clearly the case after the first and second move of Alice since,
at that point, only one or three vertices have been marked, respectively.

Suppose to the contrary that, after Bob has marked the vertex $b_i$
and Alice has marked the vertex $a_i$, $i\ge 2$,
$t$ is a threatened vertex in $G$, and that $i$ is the smallest index with this property,
which implies that $a_i$ or $b_i$ is a marked neighbor of $t$.
Thanks to Observation~\ref{obs:out-uu'}, we know that both $u$ and $u'$
have been marked. Therefore, $t$ is necessarily a child vertex of some 
trapezoid $T_{vv'}$ (we may have $vv'=uu'$).
Let $t^\ell$ and $t^r$ denote the left and right neighbors of $t$ (in $G_m$), if
they exist.
Note that at least one of them must exist, since otherwise $t$ would 
have at most two marked neighbors, and thus could not be a threatened vertex.
Since $t$ has four marked neighbors,
at least one of $t^\ell$, $t^r$ must be marked, since otherwise no vertex lying
below $t$ could have been marked, due to the connectivity constraint, so that, again,
$t$ would have at most two marked neighbors.
Thanks to Observation~\ref{obs:out-both-parents}, we thus get that both $v$ and $v'$ are marked
if $vv'$ is an edge in $G$.

We now claim that neither $T_{t^\ell t}$ nor $T_{tt^r}$ contains a marked child vertex
which is a neighbor of $t$.
Indeed, such a vertex, say $w$, cannot have been marked by Bob since, 
by Rule R1, Alice would have marked $t$ just after Bob had marked the first such child vertex 
of the corresponding trapezoid.
The vertex $w$ has thus been marked by Alice which implies, since $t$ is unmarked, that
none of the edges $tv$, $tv'$, $tt^\ell$ or $tt^r$ belong to $G$ (otherwise $t$ would have
been marked in priority by Alice), and that $w$ is the only marked neighbor of $t$, so that
$t$ cannot be a threatened vertex.

Therefore, the four marked neighbors of $w$ are necessarily $v$, $v'$, $t^\ell$ and $t^r$.
Hence, $t$ is the pivot of $T_{vv'}$, which implies, since $t$ is unmarked, that 
$t^\ell$ has been marked after $v$, and that $t^r$ has been marked after $v'$,
so that $b_i\in\{t^\ell,t^r\}$. (Note here that we cannot have $b_i\in\{v,v'\}$,
since this would imply $a_i\in\{t^\ell,t^r\}$, contradicting the priority of rule R2.)
But in each case, that is, $b_i=t^\ell$ or $b_i=t^r$, $t$ would have been marked by Alice,
thanks to Rule R3.

We thus get a contradiction in each case, which concludes the proof of Theorem~\ref{th:out-marking}.
\end{proof}

%Eric
%
%
%XXXXXXXXXXXXXXXXXXXXXXXXXXXXXXXXXX  \ES{To be done.}
%
%%
%%We consider two cases, depending on whether $b_i$ is a neighbor of $w$ or not.
%%
%%\begin{enumerate}
%%\item $b_i$ is a neighbor of $w$.\\
%%If $b_i=v$, then we cannot have $a_i=v'$ since this would imply that
%%neither $v$ nor $v'$ were marked after the $(i-1)$-th move of Alice, so that
%%$w$ has exactly two marked neighbors, a contradiction.
%%
%%%Without loss of generality, suppose that $w$ is a neighbor of $v$ and not a neighbor of $v'$.
%%
%%xxx
%%
%%\item $b_i$ is not a neighbor of $w$.\\
%%In that case, since $i$ is minimal, $a_i$ is necessarily a neighbor of $w$.
%%
%%
%%xxx
%%
%%\end{enumerate}
%%
%%
%%
%%
%%
%
%This concludes the proof of Theorem~\ref{th:out-marking}.
%\end{proof}

\medskip

Concerning the connected game chromatic number of connected outerplanar graphs, we can now  
prove the following.

\begin{theorem}\label{th:out-coloring}
If $G$ is a connected outerplanar graph, then $\chicg(G)\le 5$.
Moreover, there exist connected outerplanar graphs with $\chicg(G)=4$.
\end{theorem}

\begin{figure}
\begin{center}
\begin{tikzpicture}[x=1cm,y=1cm]
\node[scale=0.7,draw,circle,fill=white] (a) at (0,0) {};
\node[above] at (0,0.2) {$v_1$};
\node[scale=0.7,draw,circle,fill=white] (b) at (2,-1) {};
\node[right] at (2.2,-1) {$v_2$};
\node[scale=0.7,draw,circle,fill=white] (c) at (2,-3) {};
\node[right] at (2.2,-3) {$v_3$};
\node[scale=0.7,draw,circle,fill=white] (d) at (0,-4) {};
\node[below] at (0,-4.2) {$v_4$};
\node[scale=0.7,draw,circle,fill=white] (e) at (-2,-3) {};
\node[left] at (-2.2,-3) {$v_5$};
\node[scale=0.7,draw,circle,fill=white] (f) at (-2,-1) {};
\node[left] at (-2.2,-1) {$v_0$};
\draw[-] (a) -- (b);
\draw[-] (b) -- (d);
\draw[-] (a) -- (d);
\draw[-] (a) -- (e);
\draw[-] (a) -- (f);
\draw[-] (b) -- (c);
\draw[-] (c) -- (d);
\draw[-] (d) -- (e);
\draw[-] (e) -- (f);
\end{tikzpicture}
\caption{\label{fig:out-tight}An outerplanar graph with connected game chromatic number~4.}
\end{center}
\end{figure}
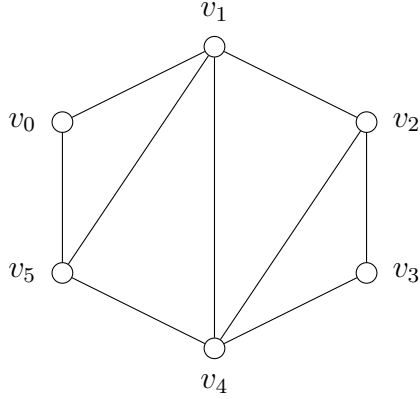

\begin{proof}
From Observation~\ref{obs:inequality-chicg} and Theorem~\ref{th:out-marking}, we get 
$\chicg(G)\le \colcg(G)\le 5$.
For the second part of the statement, consider the outerplanar graph
$G$ depicted in Figure~\ref{fig:out-tight}.
We first prove that Bob has a winning strategy when playing the connected
coloring game on $G$ with three colors. 
Thanks to the symmetries in $G$, and up to permutation of colors,
Alice has three possible first moves, that we
consider separately.
\begin{enumerate}
\item If Alice colors $v_0$ with color $1$, then Bob colors $v_1$ with color~$2$.
Now, if Alice colors $v_5$ with color~$3$ then Bob colors $v_2$ with color~$1$
so that $v_4$ is saturated, while
if Alice colors $v_2$ or $v_4$ with color~$1$ (resp. with color~$3$),
then Bob colors $v_3$ with color~$3$ (resp. with color~$1$),
so that $v_4$ or $v_2$ is saturated.
\item If Alice colors $v_1$ with color $1$, then Bob colors $v_0$ with color~$2$, 
and the so-obtained configuration is similar to that of the previous case.
\item If Alice colors $v_2$ with color $1$, then Bob colors $v_3$ with color~$2$.
Now, if Alice colors $v_4$ with color~$3$ then Bob colors $v_5$ with color~$2$
so that $v_1$ is saturated, while
if Alice colors $v_1$ with color~$2$ (note that using color~$3$ would saturate $v_4$),
then Bob colors $v_5$ with color~$3$,
so that $v_4$ is saturated.
\end{enumerate}

We thus get $\chicg(G)\ge 4$. To finish the proof, we need to show that Alice
has a winning strategy when playing the connected
coloring game on $G$ with four colors. 
Since $v_1$ and $v_4$ are the only vertices with degree at least~4 in $G$,
and since they are connected by an edge, Alice can color these two vertices
in her first two moves.
The remaining uncolored vertices can then always be colored since their degree is less
than the number of available colors. 
\end{proof}

By Observation~\ref{obs:inequality-chicg}, the second part of the statement
of Theorem~\ref{th:out-coloring}
directly implies the following.

\begin{corollary}
There exist connected outerplanar graphs with $\colcg(G)\ge 4$.
\end{corollary}

%%%%%%%%%%%%%%%%%%%%%%%%%%%%%%%%%%%%%%%%%%%%%%%%%%%%%%%%%%%%%%%%%%%%%%%%%%%%%%%%%%%%%%
%%%%%%%%%%%%%%%%%%%%%%%%%%%%%%%%%%%%%%%%%%%%%%%%%%%%%%%%%%%%%%%%%%%%%%%%%%%%%%%%%%%%%%
%%%%%%%%%%%%%%%%%%%%%%%%%%%%%%%%%%%%%%%%%%%%%%%%%%%%%%%%%%%%%%%%%%%%%%%%%%%%%%%%%%%%%%
%%%%%%%%%%%%%%%%%%%%%%%%%%%%%%%%%%%%%%%%%%%%%%%%%%%%%%%%%%%%%%%%%%%%%%%%%%%%%%%%%%%%%%
%%%%%%%%%%%%%%%%%%%%%%%%%%%%%%%%%%%%%%%%%%%%%%%%%%%%%%%%%%%%%%%%%%%%%%%%%%%%%%%%%%%%%%
%%%%%%%%%%%%%%%%%%%%%%%%%%%%%%%%%%%%%%%%%%%%%%%%%%%%%%%%%%%%%%%%%%%%%%%%%%%%%%%%%%%%%%

\section{Discussion}\label{sec:discussion}

We have introduced in this paper a connected version of the graph coloring
and graph marking games.
We have proved in particular that the connected game coloring number of every
connected outerplanar graph is at most~5, and that there exist infinitely many
connected bipartite graphs on which Alice wins the connected coloring game with two colors
but loses the game if the number of colors is at least three.

We conclude this paper by listing some open questions that should be considered
for future work.

\begin{enumerate}
%\item Do there exist outerplanar graphs with connected game coloring number~5?
%If not, this will decrease to~4 the upper bound in Theorems \ref{th:out-marking}
%and~\ref{th:out-coloring}.

\item What is the optimal upper bound on the connected game coloring number and on the connected
game chromatic number of connected outerplanar graphs? We know that both these values are either 4 or~5.

\item What is the optimal upper bound on the connected game coloring number and on the connected
game chromatic number of connected planar graphs?

\item Does there exist, for every two integers $k\ge 3$ and $p\ge 1$, 
a connected graph $G_{k,p}$ on which Alice wins the connected coloring game 
with $k$ colors, while Bob wins the game with $k+p$ colors?

\item Is the connected game coloring number a monotonic parameter, that is, is it true
that for every connected subgraph $H$ of a connected graph $G$, the inequality $\colcg(H)\le\colcg(G)$ holds?

\item Does there exist a connected graph $G$ for which $\chig(G) < \chicg(G)$? or $\colg(G) < \colcg(G)$?
(That is, is it possible that the connectivity constraint is in favour of Bob?)
\end{enumerate}

\bigskip

\noindent{\bf Acknowledgments.}
The first and third authors have been supported by the ANR-14-CE25-0006 project of the French National Research Agency,
and the fourth author by %Grant number NSFC 11571319.
Grant mumbers NSFC 11971438, ZJNSF LD19A010001 and 111 project of the Ministry of Education of China.

The work presented in this paper has been initiated while the first author was visiting LaBRI, whose hospitality
was greatly appreciated.
It has been pursued during the 9th Slovenian International Conference on Graph Theory (Bled'19), attended 
by the third and fourth authors, who warmly acknowledge the organizers for
having provided a very pleasant and inspiring atmosphere.

%%%%%%%%%%%%%%%%%%%%%%%%%%%%%%%%% BIBLIOGRAPHY

\end{document}